\documentclass[fleqn]{article}
\usepackage{amsmath}
\usepackage{amsthm}
\usepackage{amssymb}
\usepackage{enumerate}
\usepackage[pdftex]{graphics}
\usepackage[all]{xy}

\theoremstyle{plain}
\newtheorem{lem}{Lemma}
\newtheorem{thm}[lem]{Theorem}
\theoremstyle{definition}
\newtheorem{defn}{Definition}

\DeclareMathOperator{\gr}{\mathbf{gr}}
\DeclareMathOperator{\grc}{\mathbf{grc}}
\DeclareMathOperator{\End}{End}
\DeclareMathOperator{\ind}{ind}
\DeclareMathOperator{\res}{res}
\DeclareMathOperator{\ch}{ch}
\DeclareMathOperator{\sgn}{sgn}
\newcommand{\ilim}{\mathop{\varprojlim}\limits}

\usepackage[pdftex]{hyperref}

\begin{document}
\begin{center}
\Large
{\bf The algebra of local unitary invariants of identical particles}
\end{center}
\vspace*{-.1cm}
\begin{center}
P\'eter Vrana
\end{center}
\vspace*{-.4cm} \normalsize
\begin{center}
Department of Theoretical Physics, Institute of Physics\\ Budapest University of Technology and Economics\\ H-1111 Budapest, Hungary

\vspace*{.2cm}
(\today)
\end{center}

\begin{abstract}
We investigate the properties of the inverse limit of the algebras of local
unitary invariant polynomials of quantum systems containing various types of
fermionic and/or bosonic particles as the dimensions of the single particle
state spaces tend to infinity. We show that the resulting algebras are free
and present a combinatorial description of an algebraically independent
generating set in terms of graphs. These generating sets can be interpreted
as minimal sets of polynomial entanglement measures distinguishing between
states showing different nonclassical behaviour.
\end{abstract}

\section{Introduction}

One of the most important questions in quantum information theory is the
classification of states containing inequivalent types of quantum correlations.
Many results exist in the case when the parts of the composite system are
effectively distinguishable, but much less is known in the case of identical
particles. When discussing entanglement at short scales, one can expect that
indistinguishability of identical particles changes the nature of quantum
correlations in a fundamental way. Understanding such correlations and
inequivalent types of entanglement of identical particles is also of great
importance outside quantum information theory \cite{Amico}.

Following recent success in describing the algebra of polynomial invariants
under the local unitary group of multipartite quantum systems in the limiting
case of large Hilbert space dimensions \cite{Vrana}, in this paper we turn
to the question of finding a similar description when identical particles are
present.

In the case of distinguishable subsystems one has a natural notion of local
operations: these can be realized by acting on the parts separately. Lack of
distinguishability of particles, however, prevents one to apply this definition
directly to identical particles, and indeed, the notion of entanglement in this
case is not entirely settled yet and several different definitions exist
\cite{Shi,Barnum,Ichikawa,Sasaki}. In this paper we regard unitary evolutions
generated by Hamiltonians built from single-particle observables as the
analogue of local transformations as e.g. in ref. \cite{Eckert}. The state
space is then of the form $S^l(\mathcal{H})=\mathbb{S}_{(l)}\mathcal{H}$ (bosons)
or $\Lambda^l(\mathcal{H})=\mathbb{S}_{(1^l)}\mathcal{H}$ (fermions), and a local
transformation is described by the appropriate restriction of $U^{\otimes l}$
(to be denoted by $\mathbb{S}_{(l)}U$ and $\mathbb{S}_{(1^l)}U$ respectively)
for some unitary operator $U$ acting on $\mathcal{H}$.

In some cases the physical setting provides a splitting of the single particle
Hilbert space into the direct sum of orthogonal subspaces. The most natural
example is when the position of a particle can be localized in one of a few
possible regions of space having vanishing or negligible overlap. In this situtation
we may assume that the single-particle Hamiltonian leaves the orthogonal subspaces
fixed, and this further restricts the possible transformations leading to
conservation of local particle numbers.

Finally, one may have to consider a quantum mechanical system where both bosons
and fermions are present and the most general Hilbert space obtained this way
has the form
\begin{equation}
\mathbb{S}_{\lambda_1}\mathcal{H}_1\otimes\cdots\otimes\mathbb{S}_{\lambda_k}\mathcal{H}_k
\end{equation}
on which the local unitary operations $\mathbb{S}_{\lambda_1}U_1\otimes\cdots\otimes\mathbb{S}_{\lambda_k}U_k$ act
where $\lambda_j\in\{(l_j),(1^{l_j})\}$. Note that the same Hilbert space can
also be interpreted as the state space of a system of $k$ types of particles
where different types can be distinguished but particles of the same type are
identical.

The outline of this paper is as follows. In section \ref{sec:LUinv} we introduce
the object which we study, namely, the inverse limit of the algebras of invariant
polynomials under products of unitary groups starting over any irreducible
representation where the limit is taken as the dimensions tend to infinity.
Here we state that the inverse limit does not contain ``too many'' elements,
that is every element is already represented for some finite values of the
dimensions.

Section \ref{sec:symfunc} summarizes some concepts related to the representation
theory of the unitary groups and its connection to symmetric groups and their
representations.
In section \ref{sec:Hilbser} we derive a formula for the dimensions of the
homogeneous subspaces of the inverse limit in terms of induced characters.

Section \ref{sec:mat} prepares the detailed study of bosonic and fermionic
systems by collecting some properties of regular bipartite graphs and integer
stochastic matrices. In section \ref{sec:alggen} we utilize these tools in
the simplest nontrivial case: we describe invariants of pure states of a quantum system
containing any number of identical bosons.

The main result is stated in section \ref{sec:various} which deals with the most
general quantum systems containing different types of bosonic, fermionic or
distinguishable particles in any combination. We find a convenient way of encoding
elements of an algebraically independent generating set by certain graphs.
In section \ref{sec:mixed} we show that the problem of finding invariants of
mixed states can be completely reduced to the case of pure states.

Some concluding remarks follow in section \ref{sec:conclusion}. For the readers'
convenience some more technical proofs are presented outside the main text in
section \ref{sec:proofs}.

\section{The algebra of local unitary invariants}\label{sec:LUinv}

Let $k\in\mathbb{N}$, $\lambda_1,\ldots,\lambda_k$ be nonempty partitions and
$n=(n_1,\ldots,n_k)\in\mathbb{N}^k$, and let us consider the complex Hilbert
space $\mathcal{H}_{k,(\lambda_1,\ldots,\lambda_k),n}=\mathbb{S}_{\lambda_1}(\mathbb{C}^{n_1})\otimes\mathbb{S}_{\lambda_2}(\mathbb{C}^{n_2})\otimes\cdots\otimes\mathbb{S}_{\lambda_k}(\mathbb{C}^{n_k})$.
This space carries an irreducible representation of $LU_n:=U(n_1,\mathbb{C})\times\cdots\times U(n_k,\mathbb{C})$,
the group of local unitary transformations. For $n\le n'\in\mathbb{N}^k$ (with
respect to the componentwise partial order) we have the usual inclusions $\mathbb{C}^{n_j}\hookrightarrow\mathbb{C}^{n'_j}$
sending an $n$-tuple of complex numbers to the first $n'$ components which
gives rise to the inclusion $\iota_{n,n'}:\mathcal{H}_{k,(\lambda_1,\ldots,\lambda_k),n}\hookrightarrow\mathcal{H}_{k,(\lambda_1,\ldots,\lambda_k),n'}$ (obtained by tensoring the images of the former inclusions under the appropriate Schur functor).
Regarding $LU_n$ as the subgroup of $LU_{n'}$ acting on the first $n_i$ entries
in the $i$th factor and stabilizing the image of $\iota_{n,n'}$, one can see that
$\iota_{n,n'}$ is an $LU_n$-equivariant linear map.

Let $I_{k,(\lambda_1,\ldots,\lambda_k),n}$ denote the algebra of $LU_n$-invariant
polynomial functions over $\mathcal{H}_{k,(\lambda_1,\ldots,\lambda_k),n}$, regarded
as a real vector space. Polynomial functions can be identified with elements in $S(\mathcal{H}_{k,(\lambda_1,\ldots,\lambda_k),n}\oplus\mathcal{H}_{k,(\lambda_1,\ldots,\lambda_k),n}^*)$, the symmetric algebra on $\mathcal{H}_{k,(\lambda_1,\ldots,\lambda_k),n}\oplus\mathcal{H}_{k,(\lambda_1,\ldots,\lambda_k),n}^*$ on which an action of $LU_n$ is induced and we have
\begin{equation}
I_{k,(\lambda_1,\ldots,\lambda_k),n}=S(\mathcal{H}_{k,(\lambda_1,\ldots,\lambda_k),n}\oplus\mathcal{H}_{k,(\lambda_1,\ldots,\lambda_k),n}^*)^{LU_n}
\end{equation}
As in ref. \cite{Vrana}, we use a grading which is different from the usual one in a factor
of two, and call homogeneous of degree $m$ the polynomials which are of degree $m$ both
in the coefficients and their conjugates. This convention is convenient as an
invariant must have the same degree in the coefficients and their conjugates,
as seen from the fact that $e^{i\varphi}\in U(1)\simeq Z(U(n_j,\mathbb{C}))$ acts on
$S^{p}(\mathcal{H}_{k,(\lambda_1,\ldots,\lambda_k),n})\otimes S^{q}(\mathcal{H}_{k,(\lambda_1,\ldots,\lambda_k),n}^*)$ by multiplication with
$e^{i(p-q)|\lambda_j|\varphi}$.

The inclusions $\iota_{n,n'}:\mathcal{H}_{k,(\lambda_1,\ldots,\lambda_k),n}\hookrightarrow\mathcal{H}_{k,(\lambda_1,\ldots,\lambda_k),n'}$
induce morphisms of graded algebras $\varrho_{n,n'}:I_{k,(\lambda_1,\ldots,\lambda_k),n'}\to I_{k,(\lambda_1,\ldots,\lambda_k),n}$
defined by $(\varrho_{n,n'}f)(\varphi)=f(\iota_{n,n'}\varphi)$. Thus we obtain
the inverse system $((I_{k,(\lambda_1,\ldots,\lambda_k),n})_{n\in\mathbb{N}^k},(\varrho_{n,n'})_{n\le n'\in\mathbb{N}^k})$
of graded algebras, the inverse limit of which will be denoted by $I_{k,(\lambda_1,\ldots,\lambda_k)}$ and
called the algebra of LU-invariants:
\begin{equation}
I_{k,(\lambda_1,\ldots,\lambda_k)}:=\ilim_{n\in\mathbb{N}^k}I_{k,(\lambda_1,\ldots,\lambda_k),n}=\left\{(f_n)_{n\in\mathbb{N}^k}\in\prod_{n\in\mathbb{N}^k}I_{k,(\lambda_1,\ldots,\lambda_k),n}\Bigg|\forall n\le n':f_n=\varrho_{n,n'}f_{n'}\right\}
\end{equation}
Note that $I_{k,(\lambda_1,\ldots,\lambda_k),n}$ is a quotient of $I_{k,(\lambda_1,\ldots,\lambda_k)}$.

The next lemma implies that every element of $I_{k,(\lambda_1,\ldots,\lambda_k)}$
is already represented in some $I_{k,(\lambda_1,\ldots,\lambda_k),n}$.
\begin{lem}\label{lem:lowdimiso}
Let $k\in\mathbb{N}$ and $n\le n'\in\mathbb{N}$. Then the restriction of
$\varrho_{n,n'}:I_{k,n'}\to I_{k,n}$ to the subspace of elements of degree
at most $\min\{|\lambda_1|n_1,\ldots,|\lambda_k|n_k\}$ is an isomorphism.
\end{lem}
For the proof see sec. \ref{sec:proofs}.

\section{Polynomial representations of the unitary groups, symmetric functions and the characteristic map}\label{sec:symfunc}

In this section we collect some well-known facts related to the representations
of unitary and symmetric groups and their characters. For more details see e.g.
refs. \cite{FH,Macdonald}.

Recall that the isomorphism class of a polynomial representation of the unitary
group $U(n,\mathbb{C})$ is uniquely determined by its character which is a
symmetric polynomial of the eigenvalues with integer coefficients. Irreducible
representations correspond to Schur polynomials indexed by partitions of integers
into at most $n$ parts. We are interested in the large $n$ limit, therefore it
is convenient to work with the algebra $\Lambda$ of symmetric functions in infinitely many
variables. A basis of $\Lambda$ is the set $\{s_\lambda|\lambda\vdash n,n\in\mathbb{N}\}$
where $s_\lambda$ is the Schur polynomial labelled by $\lambda$ and $\lambda\vdash m$
denotes the fact that $\lambda$ is a partition of $m$.

The usual inner product on the space of class functions on the compact group
$U(n,\mathbb{C})$ translates to an inner product on $\Lambda$ defined by
\begin{equation}
\langle s_\lambda,s_\mu\rangle=\delta_{\lambda\mu}
\end{equation}
Direct sums and tensor products of the representations correspond to sums and
products of the corresponding symmetric functions, respectively.
Given two representations $\varrho_1:U(n_1,\mathbb{C})\to U(n_2,\mathbb{C})$ and
$\varrho_2:U(n_2,\mathbb{C})\to U(n_3,\mathbb{C})$, we can form their composition
$\varrho_2\circ\varrho_1:U(n_1,\mathbb{C})\to U(n_3,\mathbb{C})$ and its character
can be calculated as the plethysm of that of $\varrho_1$ with $\varrho_2$. This
operation may be defined on $\Lambda$ by $p_n[p_{n'}]=p_{nn'}$ and $f[p_n]=p_n[f]$
for $f\in\Lambda$ and extending to be an algebra endomorphism in the outer variable.

The ring of symmetric functions is also connected to the representation theory
of the symmetric groups via the characteristic map. Denoting the space of class
functions on $S_n$ by $\mathcal{R}_n$ and the $n$th homogeneous part of $\Lambda$ by
$\Lambda^n$, we define $\ch:\mathcal{R}_n\to\Lambda^n$ by
\begin{equation}
\ch f=\sum_{\mu\vdash n}z_\mu^{-1}f(\mu)p_\lambda
\end{equation}
where $f\in \mathcal{R}_n$ and $f(\mu)$ denotes its value on the conjugacy class labelled
by $\mu$, that is, on elements of cycle type $\mu$, and
\begin{equation}
z_\mu=\prod_{i}i^{a_i}i!
\end{equation}
if $\mu$ has $a_1$ ones, $a_2$ twos etc. We may also regard $\ch$ as a map
from $\bigoplus_{n\in\mathbb{N}}\mathcal{R}_n$ to $\Lambda$. Note also that $\mathcal{R}_n$
is equipped with the usual inner product, and hence also their direct sum.

Characters of irreducible representations of $S_n$ are indexed by partitions of
$n$, and will be denoted by $\chi_\mu$ where $\mu\vdash n$. These form an orthonormal
basis of $\mathcal{R}_n$, and an important property of $\ch$ is that $\ch\chi_\mu=s_\mu$,
in particular, $\ch$ is an isometry.

The sum, product and plethysm operations also have a description in terms of
representations of the symmetric group. Clearly, the direct sum of representations
corresponts to the sum in $\Lambda$. The product translates to the induction product
defined as follows. Let $\chi_i:S_{n_i}\to\mathbb{C}$ be the character of the
representation $V_i$ of $S_{n_i}$ for $i=1,2$. Then $S_{n_1}\times S_{n_2}$
can be regarded as the subgroup of $S_{n_1+n_2}$ containing bijections which
permute the first $n_1$ numbers among themselves and similarly the last $n_2$
numbers. The induction product of the two representations is
\begin{equation}
\ind_{S_{n_1}\times S_{n_2}}^{S_{n_1+n_2}}V_1\otimes V_2
\end{equation}
It can be shown that the image of its character under $\ch$ is $\ch\chi_1\cdot\ch\chi_2$

Finally, plethysm can be described in terms of the wreath product of representations.
Similarly as before, $S_k^m$ can be regarded as a subgroup of $S_{km}$. Its normalizer
is isomorphic to the wreath product $S_k\wr S_m$.

We can think of the wreath product $S_k\wr S_m$ as the set $S_k^m\times S_m$ with
multiplication defined by
\begin{equation}
(p_1,p_2,\ldots,p_m,v)\cdot(p'_1,p'_2,\ldots,p'_m,v')=(p_1p'_{v^{-1}(1)},p_2p'_{v^{-1}(2)},\ldots,p_mp'_{v^{-1}(m)},vv')
\end{equation}

Now let $V$ be a representation of $S_k$ with character $\chi$ and $W$ a representation of $S_m$ with character $\theta$. We
can define a representation of $S_k\wr S_m$ on $V^{\otimes m}\otimes W$ by
\begin{equation}
(p_1,p_2,\ldots,p_m,v)\cdot(x_1\otimes\cdots\otimes x_m\otimes y)=(p_1\cdot x_{v^{-1}(1)})\otimes\cdots\otimes (p_m\cdot x_{v^{-1}(m)})\otimes(v\cdot y))
\end{equation}
We denote the character of this representation by $\chi\wr\theta$.
Regarding $S_k\wr S_m$ as the subgroup of $S_{km}$ above, we have the following
equality:
\begin{equation}
\ch\ind_{S_k\wr S_m}^{S_{km}}(\chi\wr\theta)=(\ch\theta)[\ch\chi]
\end{equation}

Finally, applying to irreducible representations and expanding in the basis of
Schur functions (equivalently: irreducible characters of the symmetric groups)
we write:
\begin{equation}
\ind_{S_k\wr S_m}^{S_km}(\chi_\lambda\wr\chi_mu)=\sum_{\nu\vdash km}M_{\lambda\mu\nu}\chi_\nu
\end{equation}
and
\begin{equation}
s_\mu[s_\lambda]=\sum_{\nu\vdash km}M_{\lambda\mu\nu}s_\nu
\end{equation}
or as the composition of the corresponding Schur functors
\begin{equation}
\mathbb{S}_\mu\mathbb{S}_\lambda=\sum_{\nu\vdash km}M_{\lambda\mu\nu}\mathbb{S}_\nu
\end{equation}

\section{Hilbert series of the algebra of LU-invariants}\label{sec:Hilbser}

Let $k,m\in\mathbb{N}$, $\lambda_1,\ldots,\lambda_k$ be nonempty partitions and
$n\in\mathbb{N}^k$ such that $n\ge(m,m,\ldots,m)$. Our aim is to calculate the dimension $d_m$
of the $m$th graded subspace of $I_{k,(\lambda_1,\ldots,\lambda_k),n}$. This subspace
can be written as
\begin{equation}
S^{2m}(\mathcal{H}_{k,(\lambda_1,\ldots,\lambda_k),n}\oplus\mathcal{H}_{k,(\lambda_1,\ldots,\lambda_k),n}^*)^{LU_n}=(S^m(\mathcal{H}_{k,(\lambda_1,\ldots,\lambda_k),n})\oplus S^m(\mathcal{H}_{k,(\lambda_1,\ldots,\lambda_k),n})^*)^{LU_n}
\end{equation}
hence if we write $S^m(\mathcal{H}_{k,(\lambda_1,\ldots,\lambda_k),n})$ as the
direct sum of irreducible representations of $LU_n$, then $d_m$ is the sum of
the squares of multiplicities.

For $\lambda\vdash m$ we have the isomorphism
\begin{equation}
\mathbb{S}_\lambda(V_1\otimes\ldots\otimes V_k)\simeq\bigoplus_{\mu_1,\ldots,\mu_k\vdash m}C_{\lambda\mu_1\ldots\mu_k}\mathbb{S}_{\mu_1}V_1\otimes\cdots\mathbb{S}_{\mu_k}V_k
\end{equation}
where $C_{\lambda\mu_1\ldots\mu_k}=\langle\chi_\lambda,\chi_{\mu_1}\cdots\chi_{\mu_k}\rangle_{S_m}$
is the multiplicity of the irreducible representation of $S_m$ indexed by the $\lambda$ in
the tensor product of irreducible representations corresponding to $\mu_1,\ldots,\mu_k$.

Now we can write
\begin{equation}
\begin{split}
S^m(\mathcal{H}_{k,(\lambda_1,\ldots,\lambda_k),n})
 & \simeq \bigoplus_{\substack{\mu_1,\ldots,\mu_k\vdash m \\ \nu_1,\ldots,\nu_k \\ \nu_i\vdash m|\lambda_i|}}C_{(m)\mu_1\ldots\mu_k}M_{\lambda_1\mu_1\nu_1}\cdots M_{\lambda_1\mu_1\nu_1}\mathbb{S}_{\nu_1}\mathbb{C}^{n_1}\otimes\cdots\otimes\mathbb{S}_{\nu_k}\mathbb{C}^{n_k}  \\
 & \simeq \bigoplus_{\substack{\nu_1,\ldots,\nu_k \\ \nu_i\vdash m|\lambda_i|}}\left(\sum_{\mu_1,\ldots\mu_k\vdash m}C_{(m)\mu_1\ldots\mu_n}M_{\lambda_1\mu_1\nu_1}\cdots M_{\lambda_1\mu_1\nu_1}\right)\mathbb{S}_{\nu_1}\mathbb{C}^{n_1}\otimes\cdots\otimes\mathbb{S}_{\nu_k}\mathbb{C}^{n_k}
\end{split}
\end{equation}
and therefore
\begin{equation}
\begin{split}
d_m & = \sum_{\substack{\nu_1,\ldots,\nu_k \\ \nu_i\vdash m|\lambda_i|}}\left(\sum_{\mu_1,\ldots\mu_k\vdash m}C_{(m)\mu_1\ldots\mu_n}M_{\lambda_1\mu_1\nu_1}\cdots M_{\lambda_1\mu_1\nu_1}\right)^2  \\
 & = \sum_{\substack{\nu_1,\ldots,\nu_k \\ \nu_i\vdash m|\lambda_i|}}\sum_{\substack{\mu_1,\ldots\mu_k\vdash m \\ \mu'_1,\ldots\mu'_k\vdash m}}C_{(m)\mu_1\ldots\mu_n}C_{(m)\mu'_1\ldots\mu'_n}M_{\lambda_1\mu_1\nu_1}M_{\lambda_1\mu'_1\nu_1}\cdots M_{\lambda_1\mu_1\nu_1}M_{\lambda_1\mu'_1\nu_1}  \\
 & = \sum_{\substack{\mu_1,\ldots\mu_k\vdash m \\ \mu'_1,\ldots\mu'_k\vdash m}}\langle 1,\chi_{\mu_1}\cdots\chi_{\mu_k}\rangle_{S_m}\langle 1,\chi_{\mu'_1}\cdots\chi_{\mu'_k}\rangle_{S_m}\cdot \\
 &\phantom{= } \cdot\sum_{\substack{\nu_1,\ldots,\nu_k \\ \nu_i\vdash m|\lambda_i|}}\prod_{i=1}^k\langle\ind_{S_{\lambda_1}\wr S_m}^{S_{m|\lambda_i|}}\chi_{\lambda_i}\wr\chi_{\mu_1},\chi_{\nu_1}\rangle_{S_{m|\lambda_i|}}\langle\chi_{\nu_1},\ind_{S_{\lambda_1}\wr S_m}^{S_{m|\lambda_i|}}\chi_{\lambda_i}\wr\chi_{\mu'_1}\rangle_{S_{m|\lambda_i|}}  \\
 & = \sum_{\substack{\mu_1,\ldots\mu_k\vdash m \\ \mu'_1,\ldots\mu'_k\vdash m}}\Big\langle\langle 1,\chi_{\mu_1}\cdots\chi_{\mu_k}\rangle_{S_m}\ind_{H_m'}^{G_m}(\chi_{\lambda_1}\wr\chi_{\mu_1})\times\cdots\times(\chi_{\lambda_k}\wr\chi_{\mu_k}),  \\
 &\phantom{= } \langle 1,\chi_{\mu'_1}\cdots\chi_{\mu'_k}\rangle_{S_m}\ind_{H_m'}^{G_m}(\chi_{\lambda_1}\wr\chi_{\mu'_1})\times\cdots\times(\chi_{\lambda_k}\wr\chi_{\mu'_k})\Big\rangle_{S_{m|\lambda_1|}\times\cdots\times S_{m|\lambda_k|}}
\end{split}
\end{equation}
where we have used that the appearing characters are all real, irreducible characters
$\chi_{\nu_i}$ form an orthonormal basis and introduced the notation $G_m=S_{m|\lambda_1|}\times\cdots\times S_{m|\lambda_k|}$ and $H_m'=(S_{|\lambda_1|}\wr S_m)\times\cdots\times(S_{|\lambda_k|}\wr S_m)$

To simplify the last expression we need the following.
\begin{lem}\label{lem:wrindmany}
Let $A_1,\ldots A_k$ be finite groups, $m\in\mathbb{N}$, $\mu\vdash m$ and let
$\alpha_i$ be a class function on $A_i$ for every $i$. Then
\begin{equation}
\ind_{(A_1\times\cdots\times A_k)\wr S_m}^{(A_1\wr S_m)\times\cdots\times(A_k\wr S_m)}(\alpha_1\times\cdots\times\alpha_k)\wr\chi_{\mu}=\sum_{\mu_1,\ldots,\mu_k\vdash m}\langle\chi_{\mu},\chi_{\mu_1}\cdots\chi_{\mu_k}\rangle_{S_m}(\alpha_1\wr\chi_{\mu_1})\times\cdots\times(\alpha_k\wr\chi_{\mu_k})
\end{equation}
\end{lem}
The proof can be found in sec. \ref{sec:proofs}

Using this result the dimension of the degree $m$ homogeneous subspace of
$I_{k,(\lambda_1,\ldots,\lambda_k)}$ can be rewritten as
\begin{thm}
\begin{equation}\label{eq:stabdim}
d_m=\langle\ind_{H_m}^{G_m}(\chi_{\lambda_1}\times\cdots\times\chi_{\lambda_k})\wr 1,\ind_{H_m}^{G_m}(\chi_{\lambda_1}\times\cdots\times\chi_{\lambda_k})\wr 1\rangle
\end{equation}
where $G_m=S_{m|\lambda_1|}\times\cdots\times S_{m|\lambda_k|}$ and
$H_m=(S_{|\lambda_1|}\times\cdots\times S_{|\lambda_k|})\wr S_m$.
\end{thm}
Note that in the $m=0$ case $G_m=H_m$ is the trivial group, and therefore $d_0=1$
corresponding to the one dimensional space of constant polynomials.

In the special case $\lambda_1=\lambda_2=\ldots=\lambda_k=(1)$ the groups
reduce to $G_m=S_m^k$ and $H_m=S_m$, and we have that $d_m$ equals the
number of orbits of $S_m^k$ under $S_m\times S_m$ acting via left and right
multiplication, or equivalently, the number of orbits of $S_m^{k-1}$ under
$S_m$ acting by simultaneous conjugation \cite{HW}.

In the remaining sections we will restrict ourselves to the study of the
cases where the Ferrers diagrams of the appearing partitions have either
a single row or a single column, which corresponds to a multipartite quantum
system with various types of bosonic and fermionic particles. The next two
sections deal with the simplest case of a system of identical bosons.

\section{Integer stochastic matrices and regular bipartite graphs}\label{sec:mat}

Following the strategy of refs. \cite{HWW,Vrana}, in this section we introduce certain combinatorial
objects which can be used to conveniently label a basis of $I_{1,((l))}$.

An integer stochastic matrix is a square matrix with nonnegative integer entries
and such that the sum of entries in each row and in each column is the same.
Clearly, if $M$ and $N$ are two such matrices with line sums $l$, then the
block-diagonal matrix
\begin{equation}
M\oplus' N:=\left[\begin{array}{cc}
M & 0  \\
0 & N
\end{array}\right]
\end{equation}
built from them is also an integer stochastic matrix with line sums $l$.
We will call integer stochastic matrices differing only in a permutation of
rows and columns equivalent. Clearly, this is an equivalence relation and
$\oplus'$ gives rise to a well defined operation on the equivalence classes,
to be denoted by $\oplus$.

A matrix with nonnegative integer entries may also be thought of as biadjacency
matrix of a bipartite graph $G=(V_1,V_2,E)$ on a labelled vertex set, say $V_1=\{r_1,r_2,\ldots,r_m\}$ and $V_2=\{c_1,c_2,\ldots,c_{m'}\}$ (and possibly with multiple edges) together with
a fixed order of colour classes, and this correspondance is a bijection. Our
convention will be that rows correspond to vertices in $V_1$ while columns
correspond to vertices in $V_2$. Clearly, the biadjacency matrix is an $m\times m$
integer stochastic matrix iff the corresponding graph is $l$-regular and $|V_1|=|V_2|=m$.
Under this map the binary operation $\oplus$ corresponds to the disjoint union
of bipartite graphs, which we define so that we keep track of the ordering of
colour classes, i.e. for $G=(V_1,V_2,E)$ and $G'=(V_1',V_2',E')$ we have
$G\sqcup G'=(V_1\sqcup V_1',V_2\sqcup V_2',E\sqcup E')$. Note that this is
important as the bipartition of a disconnected bipartite graph is not unique.
Passing to equivalence classes of the integer stochastic matrices corresponds
to forgetting the labels of vertices.

We will make use of another description in terms of permutations. Let $p\in S_{lm}$
be a permutation of the numbers $\{1,2,\ldots,lm\}$. To $p$ we can associate an
$l$-regular $m$ by $m$ bipartite graph as follows. Let $V_1=\{r_1,r_2,\ldots,r_m\}$
and $V_2=\{c_1,c_2,\ldots,c_m\}$ be the two vertex sets, and for each $i\in\{1,2,\ldots,lm\}$
add an edge joining $r_{\lceil\frac{i}{l}\rceil}$ with $c_{\lceil\frac{p(i)}{l}\rceil}$
where $\lceil x\rceil$ denotes the smallest integer not less than $x$. In other words,
there are $e$ edges joining $r_i$ with $c_j$ iff $e$ of numbers in the range $\{(i-1)l+1,\ldots,il\}$
are mapped by $p$ into the range $\{(j-1)l+1,\ldots,jl\}$.

It is not hard to see that $p,p'\in S_{lm}$ are mapped to the same graph with
labelled vertices iff $p'=apb$ with $a,b\in S_l^m\le S_{lm}$ regarded as the subgroup
whose elements permute the numbers $\{li+1,li+2,\ldots,li+l\}$ among themselves
for each $0\le i\le m-1$. Reordering the labels of the vertices amounts to left
and right multiplication with an element of $S_m\le S_{lm}$ regarded as the subgroup
permuting the $m$ blocks of numbers $\{li+1,li+2,\ldots,li+l\}$ with $i=0,\ldots,m-1$
without reordering the numbers inside the blocks. The two subgroups generate the
normalizer of the former in $S_{lm}$, which is isomorphic to $S_l\wr S_m$. In the
following we will always assume $S_l\wr S_m$ to be this subgroup in $S_{lm}$.

To sum up, equivalence classes of $m\times m$ integer stochastic matrices with
line sum $l$ are in bijection with $l$-regular bipartite graphs with a fixed
bipartition into two $m$-element vertex sets, which in turn are in bijection
with elements of the double coset space $(S_l\wr S_m)\backslash S_{lm}/(S_l\wr S_m)$.
Fig. \ref{fig:l3m3} shows the five possible graphs in the $l=m=3$ case.

\begin{figure}[htb]
\centering
\includegraphics{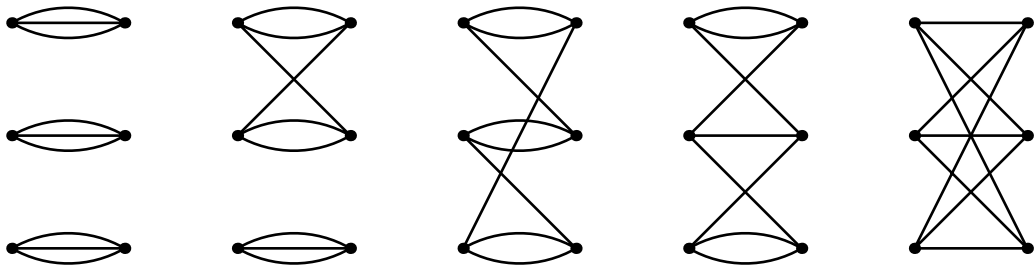}
\caption{All five $3$-regular bipartite graphs on six vertices. Representatives of the corresponding equivalence classes of permutations (from left to right) are: $(1)(2)(3)(4)(5)(6)(7)(8)(9)$, $(1)(2)(3)(4)(5)(67)(8)(9)$, $(1)(2)(34)(5)(6)(7)(8)(9)$, $(1)(2)(34)(57)(68)(9)$, $(1)(24)(37)(5)(68)(9)$. \label{fig:l3m3}}
\end{figure}

\section{Algebraically independent generators of the algebra of LU-invariants for pure states of bosons}\label{sec:alggen}

Our next aim is to prove that $I_{1,((l))}$ is free by presenting an algebraically
independent generating set.
Let $\mathcal{H}$ be a Hilbert space, and $l\in\mathbb{N}$, and consider
the $l$-boson state space $S^{l}(\mathcal{H})$. Invariant polynomial functions on
this space are in bijection with elements of the space
\begin{equation}
\begin{split}
S(S^l(\mathcal{H})\oplus S^l(\mathcal{H}^*))^{U(n,\mathbb{C})}
  & \simeq \bigoplus_{m\in\mathbb{N}}(S^m(S^l(\mathcal{H}))\otimes S^m(S^l(\mathcal{H}^*))^{U(n,\mathbb{C})}  \\
  & \simeq \bigoplus_{m\in\mathbb{N}}\End(S^m(S^l(\mathcal{H})))^{U(n,\mathbb{C})}  \\
  & \simeq \bigoplus_{m\in\mathbb{N}}\End_{U(n,\mathbb{C})}(S^m(S^l(\mathcal{H})))
\end{split}
\end{equation}
Note that $S^m(S^l(\mathcal{H}))$ can be regarded as a subspace of
$\mathcal{H}^{\otimes lm}$, and as this space comes equipped with an inner
product induced from that of $\mathcal{H}$, we also have the orthogonal
projection $\mathcal{H}^{\otimes lm}\to S^m(S^l(\mathcal{H}))$. This implies
that $\End(S^m(S^l(\mathcal{H})))$ can be viewed both as a subspace and as
a quotient of $\End(S^m(S^l(\mathcal{H})))$.

Recall that $GL(\mathcal{H})$ acts on $\mathcal{H}^{\otimes lm}$ by the defining
representation on each factor and on this space there is an action of $S_{lm}$
permuting the factors and these two actions commute. Schur-Weyl duality states
that the algebra homomorphism $\mathbb{C}S_{lm}\to\End_{GL(\mathcal{H})}(\mathcal{H}^{\otimes lm})$
is surjective, and it is injective iff $\dim \mathcal{H}\ge lm$. In these formulae
we can also write $U(\mathcal{H})$ instead of $GL(\mathcal{H})$.

Therefore we also have a surjection from the group algebra $\mathbb{C}S_{lm}$
to the space of degree $m$ polynomial invariants. Let $\{e_1,\ldots,e_n\}$ be
an orthonormal basis of $\mathcal{H}$, and $\{e_1^*,\ldots,e_n^*\}$ its dual
basis. Vectors of the form
\begin{equation}
e_{j_1}\vee e_{j_2}\vee\cdots\vee e_{j_l}:=\frac{1}{n!}\sum_{\pi\in S_l}e_{j_{\pi(1)}}\otimes e_{j_{\pi(2)}}\otimes\cdots\otimes e_{j_{\pi(l)}}
\end{equation}
with $1\le j_1\le\cdots\le j_l\le n$ form a basis of $S^l(\mathcal{H})$.
Then our surjection is defined on the group elements as
\begin{multline}
\sigma\mapsto f_{[\sigma]}:=\sum_{j_1,\ldots,j_{lm}=1}^{n}(e_{j_1}\vee\cdots\vee e_{j_k})(e_{j_{l+1}}\vee\cdots\vee e_{j_{2l}})\cdots(e_{j_{(m-1)l+1}}\vee\cdots\vee e_{j_{ml}})\cdot  \\
\cdot(e^*_{j_{\sigma(1)}}\vee\cdots\vee e^*_{j_{\sigma(l)}})\cdots(e^*_{j_{\sigma((m-1)l+1)}}\vee\cdots\vee e^*_{j_{\sigma(ml)}})
\end{multline}
and extended linearly. As the symmetrized product $\vee$ as well as the product
in $S(S^l(\mathcal{H}))$ is commutative, and we sum over every possible $lm$-tuple
of integers (in particular, over permutations of any fixed $lm$-tuple), we have
that the image of $\sigma$ under this map is the same as the image of any element
in the double coset $[\sigma]:=(S_l\wr S_m)\sigma(S_l\wr S_m)$. In other words, we have a
well-defined surjection from a vector space freely generated by the double coset
space $(S_l\wr S_m)\backslash S_{lm}/(S_l\wr S_m)$ to the space of degree $m$
invariant polynomials.

But the dimension of the two spaces can be seen to be equal, so that this surjection
is in fact an isomorphism if $n\ge lm$. Indeed, in this special case eq. (\ref{eq:stabdim}) reduces to
\begin{equation}
d_m = \langle \ind_{S_l\wr S_m}^{S_{lm}}(1\wr 1),\ind_{S_l\wr S_m}^{S_{lm}}(1\wr 1)\rangle = |(S_l\wr S_m)\backslash S_{lm}/(S_l\wr S_m)|
\end{equation}
using the well-known formula $|H\backslash G/H|=\langle\ind_H^G 1,\ind_H^G 1\rangle$ for the number of double cosets.

What we have seen so far is that a basis of $I_{1,((l))}$ can be labelled by
$l$-regular bipartite graphs or equivalently by integer stochastic matrices
with line sums $l$ up to permutation of rows and columns. It is not hard to
see (and will be proved later in a more general setting) that under these
bijections multiplication of the basis elements corresponds to the 
disjoint union and to the operation $\oplus$ defined above, respectively, implying
that any element of the basis can be uniquely written as the product of basis
elements corresponding to \emph{connected} graphs. Our results
are summarized in the following
\begin{thm}
$I_{1,((l))}$ is a free algebra, and an algebraically independent generating
set is $\{f_{s}|s\in S\}$ where $S$ is the set of elements of
$\bigsqcup_{m\in\mathbb{N}}(S_l\wr S_m)\backslash S_{lm}/(S_l\wr S_m)$
corresponding to connected graphs by the above bijection.
\end{thm}

\section{Systems of different kinds of indistinguishable particles}\label{sec:various}

We turn to the local unitary invariants of the most general quantum systems
containing various types of bosonic, fermionic and distinguishable particles.

The starting point is the stable dimension formula of eq. \ref{eq:stabdim}.
In the present case $k$ is arbitrary while the partitions $\lambda_1,\ldots,\lambda_k$
are either of the form $(l_i)$ or $(1^{l_i})$ corresponding to bosons or fermions
(so that $l_i=|\lambda_i|$).

Let $G_m=S_{ml_1}\times\cdots\times S_{ml_k}$ and $H_m=(S_{l_1}\times\cdots\times S_{l_k})\wr S_m\le G_m$
as before. The degree $m$ subspace of $I_{k,(\lambda_1,\ldots,\lambda_k)}$ is
\begin{equation}\label{eq:dmcosets}
\begin{split}
d_m
 & = \langle\ind_{H_m}^{G_m}(\chi_{\lambda_1}\times\cdots\times\chi_{\lambda_k})\wr 1,\ind_{H_m}^{G_m}(\chi_{\lambda_1}\times\cdots\times\chi_{\lambda_k})\wr 1\rangle_{G_m}  \\
 & = \langle\res_{G_m}^{H_m}\ind_{H_m}^{G_m}(\chi_{\lambda_1}\times\cdots\times\chi_{\lambda_k})\wr 1,(\chi_{\lambda_1}\times\cdots\times\chi_{\lambda_k})\wr 1\rangle_{H_m}  \\
 & = \sum_{s\in {H_m}\backslash {G_m}/{H_m}}\langle((\chi_{\lambda_1}\times\cdots\times\chi_{\lambda_k})\wr 1)^s,\res_{H_m}^{{H_m}_s}(\chi_{\lambda_1}\times\cdots\times\chi_{\lambda_k})\wr 1\rangle_{{H_m}_s}
\end{split}
\end{equation}
by Mackey's theorem\cite{Simon} where ${H_m}_s={H_m}\cap s{H_m}s^{-1}$ and $f^s(h)=f(s^{-1}hs)$ and in
the sum $s$ runs over a system of representatives of the set of double cosets.
Note that the inner product indeed depends on $s$ only through the double coset
${H_m}s{H_m}$.

Observe that each term in the sum is either $0$ or $1$, being the inner product
of one dimensional and hence irreducible characters, therefore
\begin{equation}\label{eq:dmbound}
d_m\le|{H_m}\backslash {G_m}/{H_m}|
\end{equation}
holds. In the following we will give a set of $|{H_m}\backslash {G_m}/{H_m}|$ invariants
spanning the degree $m$ homogeneous subspace, and we shall see that when
strict inequality holds in eq. (\ref{eq:dmbound}), some invariants become zero while
the remaining ones form a basis.

Let $n=(n_1,\ldots,n_k)$ be a fixed $k$-tuple of integers such that $n_i\ge ml_i$.
First observe that $S^m(\mathbb{S}_{\lambda_1}(\mathbb{C}^{n_1})\otimes\cdots\otimes\mathbb{S}_{\lambda_k}(\mathbb{C}^{n_k}))$
can be thought of both as a subrepresentation and as a quotient of $(\mathbb{C}^{n_1})^{\otimes ml_1}\otimes\cdots\otimes(\mathbb{C}^{n_k})^{\otimes ml_k}$ via symmetrization/antisymmetrization. Hence we have a surjection
\begin{equation}
\End((\mathbb{C}^{n_1})^{\otimes ml_1}\otimes\cdots\otimes(\mathbb{C}^{n_k})^{\otimes ml_k})\to\End(S^m(\mathbb{S}_{\lambda_1}(\mathbb{C}^{n_1})\otimes\cdots\otimes\mathbb{S}_{\lambda_k}(\mathbb{C}^{n_k})))
\end{equation}
which is $LU_n$-equivariant, therefore we also have a surjection
\begin{equation}
(\End((\mathbb{C}^{n_1})^{\otimes ml_1}\otimes\cdots\otimes(\mathbb{C}^{n_k})^{\otimes ml_k}))^{LU_n}\to(\End(S^m(\mathbb{S}_{\lambda_1}(\mathbb{C}^{n_1})\otimes\cdots\otimes\mathbb{S}_{\lambda_k}(\mathbb{C}^{n_k}))))^{LU_n}
\end{equation}
Here the right hand side can be identified with the degree $m$ homogeneous subspace
of $I_{k,(\lambda_1,\ldots,\lambda_k)}$.

Let $e_1,e_2,\ldots$ denote the standard basis of $\mathbb{C}^n$ and let us introduce
the following notation:
\begin{equation}
e_{\lambda,i_1i_2\ldots i_l}=\frac{1}{l!}\sum_{\pi\in S_l}\chi_\lambda(\pi)e_{i_{\pi(1)}}\otimes e_{i_{\pi(2)}}\otimes\cdots\otimes e_{i_{\pi(l)}}
\end{equation}
so that $\lambda=(l)$ ($\lambda=(1^l)$) corresponds to symmetrized (antisymmetrized)
tensor products, respectively. Elements of the dual bases will be denoted by
$e^*_1,e^*_2,e^*_{\lambda,i_1i_2\ldots i_l}$ etc.

By Schur-Weyl duality we have the surjection
\begin{equation}
\mathbb{C}(S_{ml_1}\times\cdots\times S_{ml_k})\to(\End((\mathbb{C}^{n_1})^{\otimes ml_1}\otimes\cdots\otimes(\mathbb{C}^{n_k})^{\otimes ml_k}))^{LU_n}
\end{equation}
defined by
\begin{multline}\label{eq:tensorinv}
(\sigma_1,\ldots,\sigma_k)
 \mapsto\sum_{\substack{i_1^1,\ldots,i_1^{ml_1} \\ \vdots \\ i_k^1,\ldots,i_k^{ml_k}}}(e_{i_1^1}\otimes e_{i_1^2}\otimes\cdots\otimes e_{i_1^{ml_1}}\otimes\cdots\otimes e_{i_k^1}\otimes e_{i_k^2}\otimes\cdots\otimes e_{i_k^{ml_k}})\otimes  \\
 \otimes(e^*_{i_1^{\sigma_1(1)}}\otimes e^*_{i_1^{\sigma_1(2)}}\otimes\cdots\otimes e^*_{i_1^{\sigma_1(ml_1)}}\otimes\cdots\otimes e^*_{i_k^{\sigma_k(1)}}\otimes e^*_{i_k^{\sigma_k(2)}}\otimes\cdots\otimes e^*_{i_k^{\sigma_k(ml_k)}})
\end{multline}
on $k$-tuples of permutations and extended linearly. Composing this with the
surjection above we have the following map:
\begin{multline}\label{eq:fdef}
(\sigma_1,\ldots,\sigma_k)\mapsto f_{[\sigma_1,\ldots,\sigma_k]}:=\sum(e_{\lambda_1,i_1^1i_1^2\ldots i_1^{l_1}}\otimes e_{\lambda_2,i_2^1i_2^2\ldots i_2^{l_2}}\otimes\cdots\otimes e_{\lambda_k,i_k^1i_k^2\ldots i_k^{l_k}})\cdots \\
(e_{\lambda_1,i_1^{(m-1)l_1+1}i_1^{(m-1)l_1+2}\ldots i_1^{ml_1}}\otimes \cdots\otimes e_{\lambda_k,i_k^{(m-1)l_k+1}i_k^{(m-1)l_k+2}\ldots i_k^{ml_k}})\cdot  \\
(e^*_{\lambda_1,i_1^{\sigma_1(1)}i_1^{\sigma_1(2)}\ldots i_1^{\sigma_1(l_1)}}\otimes e^*_{\lambda_2,i_2^{\sigma_2(1)}i_2^{\sigma_2(2)}\ldots i_2^{\sigma_2(l_2)}}\otimes\cdots\otimes e^*_{\lambda_k,i_k^{\sigma_k(1)}i_k^{\sigma_k(2)}\ldots i_k^{\sigma_k(l_k)}})\cdots \\
(e^*_{\lambda_1,i_1^{\sigma_1((m-1)l_1+1)}i_1^{\sigma_1((m-1)l_1+2)}\ldots i_1^{\sigma_1(ml_1)}}\otimes \cdots\otimes e^*_{\lambda_k,i_k^{\sigma_k((m-1)l_k+1)}i_k^{\sigma_k((m-1)l_k+2)}\ldots i_k^{\sigma_k(ml_k)}})
\end{multline}
where the sum is over all possible values of the indices: $1\le i_j^1,\ldots,i_j^{ml_j}\le n_j$
($\forall j\in\{1,\ldots,k\}$) and $[\sigma_1,\ldots,\sigma_k]$ denotes the
double coset $H_m(\sigma_1,\ldots,\sigma_k)H_m$.

Note that $e_{\lambda,i_1i_2\ldots i_l}$ does not change (up to sign) upon changing
the order of the appearing indices, and also the factors in each term can be reordered
in an arbirary way. Taking into account that we sum over every possible value of the
indices, the first transformation is realized when any of the $\sigma_j$ is
multiplied from either side by a permutation in the Young subgroup $S_{l_i}\times\cdots\times S_{l_i}\le S_{ml_i}$
while the latter amounts to simultaneous multiplication of the permutations $\sigma_1,\ldots,\sigma_k$
from the left or from the right by the subgroup of $S_{ml_1}\times\cdots\times S_{ml_k}$ isomorphic
to $S_m$ which permutes simultaneously the blocks fixed by these Young subgroups. But
the subgroup generated by these is precisely ${H_m}$ so that indeed we have an up-to-sign
well defined invariant for each element of ${H_m}\backslash {G_m}/{H_m}$. The sign can be fixed by
requiring the invariant to be positive for separable states.

This already implies that when in eq. (\ref{eq:dmbound}) equality holds, a
basis of the degree $m$ homogeneous subpace is obtained this way. Similarly to
the special case considered in sec. \ref{sec:alggen} we can describe ${H_m}\backslash {G_m}/{H_m}$
in a purely combinatorial way in terms of certain graphs as follows.

Let $(\sigma_1,\ldots,\sigma_k)\in {G_m}=S_{ml_1}\times\cdots\times S_{ml_k}$. Now we
can draw a bipartite graph with vertices ${r_1,r_2,\ldots,r_m,c_1,c_2,\ldots,c_m}$
and with edges of $k$ different colours, adding an edge of the $j$th colour connecting
$r_{\lceil\frac{i}{l_j}\rceil}$ with $c_{\lceil\frac{\sigma_j(i)}{l_j}\rceil}$ for
every $1\le i\le ml_j$. In other words, there are $e$ edges of colour $j$ joining
$r_i$ with $c_{i'}$ iff $e$ of numbers in the range $\{(i-1)l_j+1,\ldots,il_j\}$
are mapped by $\sigma_j$ into the range $\{(i'-1)l_j+1,\ldots,i'l_j\}$. Finally we
forget the labels of the vertices but keep the order of the colour classes, and
this way obtain a bijection between ${H_m}\backslash {G_m}/{H_m}$ and the set of
isomorphism classes of bipartite graphs with a fixed bipartition into two $m$-element
vertex sets and edges of $k$ different colours with the subgraph given by edges
of colour $j$ being $l_j$-regular. The set of these isomorphism classes will be
denoted by $\gr(k,(l_1,\ldots,l_k),m)$, and connected ones by $\grc(k,(l_1,\ldots,l_k),m)$

As an illustration fig. \ref{fig:bipartite} shows the graph obtained in the $k=2$,
$l_1=2$, $l_2=3$, $m=3$ case from the pair of permutations $((123564),(17896)(2)(3)(4)(5))\in S_6\times S_9$.
\begin{figure}[htb]
\centering
\includegraphics{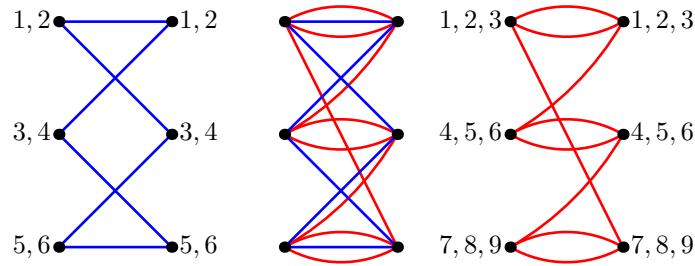}
\caption{The regular bipartite graph obtained from $(123564)\in S_6$ (left), $(17896)(2)(3)(4)(5)\in S_9$ (right) and the combined graph encoding the equivalence class of the pair in the $k=2$, $l_1=2$, $l_2=3$, $m=3$ case. (Colour online)\label{fig:bipartite}}
\end{figure}
Note also that the $l_1=l_2=\ldots=l_k=1$ special case corresponds to distinguishable
particles, and the resulting graphs can be identified (after merging pairs of vertices
along one of the colours) with the graph coverings of ref. \cite{HWW}.

Alternatively, we can identify the double cosets with orbits of $k$-tuples of
$m\times m$ integer stochastic matrices such that the $j$th matrix has line
sums $l_j$, under simultaneous multiplication from the left or from the right
by permutation matrices.

The next step is to understand multiplication in $I_{k,\lambda_1,\ldots,\lambda_k}$
in terms of graphs, i.e. to find the structure constants with respect to the
basis given above.
\begin{defn}
Let $\gr(k,(l_1,\ldots,l_k))$ denote the disjoint union
$\bigsqcup_{m\in\mathbb{N}}\gr(k,(l_1,\ldots,l_k),m)$.
This set comes equipped with an operation $\sqcup$ induced by the disjoint union of graphs.
Similarly, let $\grc(k,(l_1,\ldots,l_k))$ denote the disjoint union
$\bigsqcup_{m\in\mathbb{N}}\grc(k,(l_1,\ldots,l_k),m)$.

Let $\star$ denote the operation defined on $\bigsqcup_{m\in\mathbb{N}}H_m\backslash G_m/H_m$
as follows. For $(\sigma_1,\ldots,\sigma_k)\in G_m$ and $(\sigma'_1,\ldots,\sigma'_k)\in G_{m'}$
the usual inclusions $S_{ml_i}\times S_{m'l_i} \hookrightarrow S_{ml_i+m'l_i}$
of Young subgroups determine an element $(\pi_1,\ldots,\pi_k)$ in $G_{m+m'}$. We
define $[\sigma_1,\ldots,\sigma_k]\star[\sigma'_1,\ldots,\sigma'_k]$
to be the double coset $[\pi_1,\ldots,\pi_k]$.
\end{defn}
It is easy to see that this operation is well defined, and turns $\bigsqcup_{n\in\mathbb{N}}H_m\backslash G_m/H_m$
into a commutative monoid (the identity being the only element of $H_0\backslash G_0/H_0$).
In addition, the map $\varphi:\bigsqcup_{m\in\mathbb{N}}H_m\backslash G_m/H_m\to\gr(k,(l_1,\ldots,l_k))$
described above becomes this way an isomorphism of monoids. As $\gr(k,(l_1,\ldots,l_k))$
is freely generated by the subset $\grc(k,(l_1,\ldots,l_k))$,
the former is also freely generated by preimages of connected graphs.

\begin{lem}\label{lem:product}
Let $(\sigma_1,\ldots,\sigma_k)\in G_m$ and $(\sigma'_1,\ldots,\sigma'_k)\in G_{m'}$ be
arbitrary elements. Then
\begin{equation}
f_{[\sigma_1,\ldots,\sigma_k]}\cdot f_{[\sigma'_1,\ldots,\sigma'_k]}=f_{[\sigma_1,\ldots,\sigma_k]\star[\sigma'_1,\ldots,\sigma'_k]}
\end{equation}
\end{lem}
The proof can be found in sec. \ref{sec:proofs}.
It follows that we have a surjective homomorphism from the semigroup
algebra $\mathbb{C}\gr(k,(l_1,\ldots,l_k))\to I_{k,(\lambda_1,\ldots,\lambda_k)}$
defined by $G\mapsto f_{\varphi^{-1}(G)}$ and extended linearly. When equality
holds in eq. (\ref{eq:dmbound}), this is an isomorphism.

We turn now to the case when we have a strict inequality in eq. (\ref{eq:dmbound}).
We will see that this means only a minor modification in the structure of the
inverse limit, namely, to each \emph{nonzero} term in eq. (\ref{eq:dmcosets})
corresponds an invariant via eq. (\ref{eq:fdef}) forming a basis of the
degree $m$ homogeneous subspace, and for each vanishing term the corresponding
invariant is identically zero.

\begin{lem}\label{lem:vanishing}
Let ${G_m}$ and ${H_m}$ as above and $s=(\sigma_1,\ldots,\sigma_k)\in {G_m}$ such that
\begin{equation}
\langle((\chi_{\lambda_1}\times\cdots\times\chi_{\lambda_k})\wr 1)^s,\res_{H_m}^{{H_m}_s}(\chi_{\lambda_1}\times\cdots\times\chi_{\lambda_k})\wr 1\rangle_{{H_m}_s}=0
\end{equation} holds. Then $f_{[\sigma_1,\ldots,\sigma_k]}=0$.
\end{lem}
For the proof see sec. \ref{sec:proofs}

Now we have everything at hand to state the following:
\begin{thm}
Let $B:=\{G\in\gr(k,(l_1,\ldots,l_k))|f_{\varphi^{-1}(G)}\neq 0\}$
and let $S$ denote the subset $\{G\in\grc(k,(l_1,\ldots,l_k))|f_{\varphi^{-1}(G)}\neq 0\}$ of $B$.
Then
\begin{equation}
\{f_{\varphi^{-1}(G)}|G\in B\}
\end{equation}
is a basis of $I_{k,(\lambda_1,\ldots,\lambda_k)}$
and $I_{k,(\lambda_1,\ldots,\lambda_k)}$ is freely generated as an algebra by the
set
\begin{equation}
\{f_{\varphi^{-1}(G)}|G\in S\}
\end{equation}
\end{thm}
\begin{proof}
We have seen that $\{f_{\varphi^{-1}(G)}|G\in \gr(k,(l_1,\ldots,l_k))\}$ generates $I_{k,(\lambda_1,\ldots,\lambda_k)}$
as a vector space. Removing the zero vector does not change this property, hence $\{f_{\varphi^{-1}(G)}|G\in B\}$
is also a generating set. The number of degree $m$ elements in $\{f_{\varphi^{-1}(G)}|G\in B\}$ is at most
the number of double cosets $[s]\in H_m\backslash G_m/H_m$ such that
\begin{equation}
\langle((\chi_{\lambda_1}\times\cdots\times\chi_{\lambda_k})\wr 1)^s,\res_{H_m}^{{H_m}_s}(\chi_{\lambda_1}\times\cdots\times\chi_{\lambda_k})\wr 1\rangle_{{H_m}_s}\neq 0
\end{equation}
and this number is equal to $d_m$ by eq. (\ref{eq:dmcosets}), therefore must form a basis.

An element $G$ of $B$ can be uniquely written as the disjoint union of connected graphs
i.e. elements of $S$, because an invariant corresponding to a connected graph outside
$S$ is zero. By lemma \ref{lem:product} this means that $f_{\varphi^{-1}(G)}$ is in a unique
way the product of elements of $\{f_{\varphi^{-1}(G)}|G\in S\}$.
\end{proof}

It remains to settle the question whether we have equality or not in eq. (\ref{eq:dmbound})
for a given $k$-tuple $(\lambda_1,\ldots,\lambda_k)$. Equivalently, we wish to determine
if there exist an element $s$ in $G_m$ for some $m$ such that
\begin{equation}
\langle((\chi_{\lambda_1}\times\cdots\times\chi_{\lambda_k})\wr 1)^s,\res_{H_m}^{{H_m}_s}(\chi_{\lambda_1}\times\cdots\times\chi_{\lambda_k})\wr 1\rangle_{{H_m}_s}=0
\end{equation}
Clearly this cannot happen if $(\chi_{\lambda_1}\times\cdots\times\chi_{\lambda_k})\wr 1$
is the restriction of some character of $G_m$. We may assume that $m\ge 2$ since for
$m=0$ and $m=1$ we have $G_m=H_m={H_m}_s$.

In the $m\ge 2$ case $G_m=S_{ml_1}\times\cdots\times S_{ml_2}$ has exactly $2^k$
distinct one dimensional characters: we can chose the trivial or the alternating
character for each factor. Let $X=(X_1,\ldots,X_k)\in\{triv,alt\}^k$ be such a
choice. We wish to find out its value on an element $(\sigma_1^1,\ldots,\sigma_1^m,\sigma_2^1,\ldots,\sigma_2^m,\ldots,\sigma_k^1,\ldots,\sigma_k^m,\pi)$ of $H_m=(S_{l_1}\times\cdots\times S_{l_2})\wr S_m\le G_m$.
In the image of this element in $G_m$ the $j$th factor is a permutation of
$m$ blocks of size $l_j$ determined by $\pi\in S_m$ and inside the blocks the
permutations $\sigma_j^1,\ldots,\sigma_j^m$ act. We denote this permutation by
$\sigma_j^\pi$. With this notation we have
\begin{equation}
X_j(\sigma_j^\pi)=\left\{\begin{array}{ll}
1 & \textrm{if $X_j=triv$}  \\
\prod_{i=1}^m \sgn(\sigma_j^i) & \textrm{if $X_j=alt$, $|l_j|\equiv 0\pmod{2}$}  \\
\sgn(\pi)\prod_{i=1}^m \sgn(\sigma_j^i) & \textrm{if $X_j=alt$, $|l_j|\equiv 1\pmod{2}$}  \\
\end{array}\right.
\end{equation}
where $\sgn$ is the sign of the permutation and therefore the value of the
character of $G_m$ we are looking for is
\begin{equation}
\prod_{j=1}^k X_j(\sigma_j^\pi)=(\sgn(\pi))^{\{j\in\{1,\ldots,k\}|X_j=alt,l_j\equiv 1\pmod{2}\}}\prod_{j=1}^k\prod_{i=1}^m X_j(\sigma_j^i)
\end{equation}

Let $Y_j=triv$ if $\lambda_j=(l_j)$ and $Y_j=alt$ if $\lambda_j=(1^{l_j})$.
In the special case $l_j=1$ the two representations are the same, and we can
safely choose any of $\{triv,alt\}$ for $Y_j$. This freedom will be used shortly.
The value of $(\chi_{\lambda_1}\times\cdots\times\chi_{\lambda_k})\wr 1$ on
$(\sigma_1^1,\ldots,\sigma_1^m,\ldots,\sigma_k^1,\ldots,\sigma_k^m,\pi)$ is
\begin{equation}
\prod_{j=1}^k\prod_{i=1}^m Y_j(\sigma_j^i)
\end{equation}
The two values coincide for every element iff $\{j\in\{1,\ldots,k\}|X_j=alt,l_j\equiv 1\pmod{2}\}$
is even and for all $j\in\{1,\ldots,k\}$ either $l_j=1$ or $X_j=Y_j$. Note that
when for at least one $j$ we have $l_j=1$ then we can always choose $X$ so that
$\{|j\in\{1,\ldots,k\}|X_j=alt,l_j\equiv 1\pmod{2}\}$ is even. If $l_j>1$ for all
$j\in\{1,\ldots,k\}$ then the condition means that the total number of fermions
is even.

On the other hand, when $\forall j:l_j>1$ and the total number of fermions is
odd and $m\ge 3$ we can always find an element $s=(\sigma_1,\ldots,\sigma_k)\in G_m$ such that
\begin{equation}
\langle((\chi_{\lambda_1}\times\cdots\times\chi_{\lambda_k})\wr 1)^s,\res_{H_m}^{{H_m}_s}(\chi_{\lambda_1}\times\cdots\times\chi_{\lambda_k})\wr 1\rangle_{{H_m}_s}=0
\end{equation}

As an important special case we list some values of $d_m$ for $I_{1,((l))}$
and $I_{1,((1^l))}$, that is, for a system of $l$ bosons and fermions, respectively in table \ref{tab:bosonfermion}.
The numbers of homogeneous invariants in an algebraically independent generating
set are listed in \ref{tab:bosonfermionfree}.

Note that in the case of two particles (either bosons or fermions) $d_m$ is
the number of partitions of $m$. Accordingly, $I_{1,((2))}$ and $I_{1,((1^2))}$
is freely generated by traces of nonnegative integer powers of the one-particle
reduced density matrix.

Note that bipartite entanglement measures introduced previously for bosons\cite{Paskauskas}
and fermions\cite{Schliemann} can be expressed with invariants of the reduced
density matrix.

\begin{table}
\centering
\begin{tabular}{r|rrrrr}
  & 1 & 2 & 3 & 4 & 5 \\
\hline
1 & 1 & 1 & 1 & 1 & 1 \\
2 & 1 & 2 & 2 & 3 & 3 \\
3 & 1 & 3 & 5 & 9 & 13 \\
4 & 1 & 5 & 12 & 43 & 106 \\
5 & 1 & 7 & 31 & 264 & 1856 \\
6 & 1 & 11 & 103 & 2804 & 65481 \\
7 & 1 & 15 & 383 & 44524 & 3925518
\end{tabular}
\begin{tabular}{r|rrrrr}
  & 1 & 2 & 3 & 4 & 5 \\
\hline
1 & 1 & 1 & 1 & 1 & 1 \\
2 & 1 & 2 & 2 & 3 & 3 \\
3 & 1 & 3 & 4 & 9 & 12 \\
4 & 1 & 5 & 10 & 43 & 94 \\
5 & 1 & 7 & 23 & 264 & 1613 \\
6 & 1 & 11 & 71 & 2804 & 58793 \\
7 & 1 & 15 & 251 & 44524 & 3624974
\end{tabular}
\caption{Stable dimensions of homogeneous subspaces of the algebra of local unitary invariants of bosons (left) and fermions (right). The number of particles grows from left to right, while degree grows downwards.\label{tab:bosonfermion}}
\end{table}

\begin{table}
\centering
\begin{tabular}{r|rrrrr}
  & 1 & 2 & 3 & 4 & 5 \\
\hline
1 & 1 & 1 & 1 & 1 & 1 \\
2 & 0 & 1 & 1 & 2 & 2 \\
3 & 0 & 1 & 3 & 6 & 10 \\
4 & 0 & 1 & 6 & 31 & 90 \\
5 & 0 & 1 & 16 & 209 & 1730 \\
6 & 0 & 1 & 59 & 2453 & 63386 \\
7 & 0 & 1 & 243 & 41098 & 3855647
\end{tabular}
\begin{tabular}{r|rrrrr}
  & 1 & 2 & 3 & 4 & 5 \\
\hline
1 & 1 & 1 & 1 & 1 & 1 \\
2 & 0 & 1 & 1 & 2 & 2 \\
3 & 0 & 1 & 2 & 6 & 9 \\
4 & 0 & 1 & 5 & 31 & 79 \\
5 & 0 & 1 & 11 & 209 & 1501 \\
6 & 0 & 1 & 39 & 2453 & 56973 \\
7 & 0 & 1 & 157 & 41098 & 3562441
\end{tabular}
\caption{Numbers of free generators of the algebra of local unitary invariants of bosons (left) and fermions (right). The number of particles grows from left to right, while degree grows downwards.\label{tab:bosonfermionfree}}
\end{table}

\section{Invariants of mixed states}\label{sec:mixed}

The case of mixed state invariants can be reduced to the results of the previous
section using the same method as in ref. \cite{Albeverio,Vrana}. For $k\in\mathbb{N}$, partitions
$\lambda_1,\ldots,\lambda_k$ and dimensions $n=(n_1,\ldots,n_k)$ we have the
isomorphism
\begin{equation}
\begin{split}
I^{\mathrm{mixed}}_{k,(\lambda_1,\ldots,\lambda_k),n}
 & :=S(\End(\mathcal{H}_{k,(\lambda_1,\ldots,\lambda_k),n}))^{LU_n} \\
 & \simeq S(\mathcal{H}_{k+1,(\lambda_1,\ldots,\lambda_k,(1)),(n_1,\ldots,n_k,n_E)}\oplus\mathcal{H}_{k+1,(\lambda_1,\ldots,\lambda_k,(1)),(n_1,\ldots,n_k,n_E)}^*)^{LU_{(n_1,\ldots,n_k,n_E)}}
\end{split}
\end{equation}
for large enough $n_E$. We can think of the last subsystem (``environment'') as
the purifying system of mixed states over $\mathcal{H}_{k,(\lambda_1,\ldots,\lambda_k),n}$.
We will continue to consider only bosons and fermions, i.e. we assume that $\lambda_j$
has a single row or column for all $j$. The extra subsystem is described by the
representation of $U(n_E,\mathbb{C})$ corresponding to the partition $(1)$.
In particular, in this case we always have equality in eq. (\ref{eq:dmbound}).
The first few stable dimensions for the $k=1$ case are indicated in table (\ref{tab:mixed})
while the numbers of homogeneous invariants in an algebraically independent generating
set are listed in \ref{tab:mixedfree}.

\begin{table}
\centering
\begin{tabular}{r|rrrrr}
  & 1 & 2 & 3 & 4  \\
\hline
1 & 1 & 1 & 1 & 1  \\
2 & 2 & 3 & 4 & 5  \\
3 & 3 & 8 & 16 & 31  \\
4 & 5 & 25 & 118 & 501  \\
5 & 7 & 85 & 1411 & 19158  \\
6 & 11 & 397 & 30335 & 1468699  \\
7 & 15 & 2183 & 939789 & 186406186 
\end{tabular}
\caption{Stable dimensions of homogeneous subspaces of the algebra of local unitary invariants of mixed states of bosons or fermions. The number of particles grows from left to right, while degree grows downwards.\label{tab:mixed}}
\end{table}

\begin{table}
\centering
\begin{tabular}{r|rrrrr}
  & 1 & 2 & 3 & 4  \\
\hline
1 & 1 & 1 & 1 & 1  \\
2 & 1 & 2 & 3 & 4  \\
3 & 1 & 5 & 12 & 26  \\
4 & 1 & 14 & 96 & 460  \\
5 & 1 & 50 & 1257 & 18553  \\
6 & 1 & 265 & 28568 & 1447330  \\
7 & 1 & 1601 & 904439 & 184851055 
\end{tabular}
\caption{Numbers of free generators of the algebra of local unitary invariants of mixed states of bosons or fermions. The number of particles grows from left to right, while degree grows downwards.\label{tab:mixedfree}}
\end{table}

The fact that we have a distinguished subsystem of a single particle makes it
possible to give an alternative description of the graphs labelling invariants
in the basis or in the algebraically independent generating set given above.

Let $l_j=|\lambda_j|$ as before. Applying the results of the previous section
we have that $\gr(k+1,(l_1,\ldots,l_k,1))$ encodes elements of a basis of
$I^{\mathrm{mixed}}_{k,(\lambda_1,\ldots,\lambda_k),n}$ while the subset
$\gr(k+1,(l_1,\ldots,l_k,1))$ corresponds to a set of free generators. Recall
that by definition the edges having the $k+1$th colour give a subgraph which
is bipartite and $1$-regular, therefore we have a distinguished bijection
between the two colour classes of vertices. We can use this bijection to
identify the endpoints of such edges, and at the same time, in order to be
able to recover the original graph we direct first the remaining edges from the
first colour class to the second one. In this way $\gr(k+1,(l_1,\ldots,l_k,1))$
can be identified with the set of equivalence classes of finite directed graphs
edges of $k$ different colours such that the subgraph determined by the $j$th
colour is $l_j$-regular (i.e. at each vertex the indegree and the outdegree are
both $l_j$).

This alternative description is illustrated in fig. (\ref{fig:mixed}) for a degree $3$
invariant of a mixed state of two bosons or two fermions.
\begin{figure}[htb]
\centering
\includegraphics{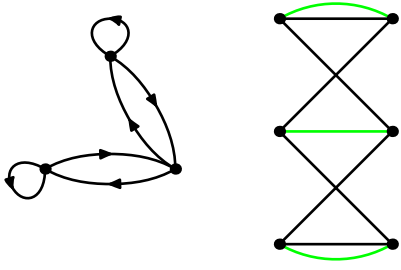}
\caption{Example of a graph corresponding to a mixed state invariant (left) obtained by identifying vertices of a graph of a pure state invariant (right) along the green edges coming from the purifying system (Colour online)\label{fig:mixed}}
\end{figure}

\section{Conclusion}\label{sec:conclusion}

In this paper we have studied the algebras of real polynomial invariants under
the local unitary groups $U(n_1,\mathbb{C}^{n_1})\times\cdots\times U(n_k,\mathbb{C}^{n_k})$
over the representation $(\mathbb{S}_{\lambda_1}\mathbb{C}^{n_1})\otimes\cdots\otimes(\mathbb{S}_{\lambda_k}\mathbb{C}^{n_k})$.
This can be interpreted as the state space of a composite quantum mechanical system
containing various types of identical particles, possibly obeying non-abelian statistics.

The algebras have the property that as $n_j\to\infty$ for all $j$, the dimension
of the homogeneous subspaces stabilize, making it convenient to work with the inverse
limit $I_{k,(\lambda_1,\ldots,\lambda_k)}$ of the algebras (in the category of graded algebras). The stable dimension
is then the dimension of the corresponding homogeneous subspace of the inverse limit,
for which a formula in terms of induced characters was derived.

In the most important case when only bosonic and fermionic particles are present,
the bound (\ref{eq:dmbound}) on the stable dimension was estabilished, which
has a combinatorial interpretation in terms of the number of graphs with a certain
property. Curiously the bound is saturated iff the total number of fermions is even.

We would like to remark that in general eq. (\ref{eq:dmbound}) does not hold.
For example in the simplest case not covered in sec. \ref{sec:various} we have $k=1$
and $\lambda_1=(2,1)$. In this case the sequence $d_m$ for $I_{1,((2,1))}$ starts
as $1,4,18,151,1628,24164,431401,\ldots$ which is to be compared with the $l=3$
column of table \ref{tab:bosonfermion}. The simplest such example in the mixed
case is $I^{\mathrm{mixed}}_{1,((2,1))}\simeq I_{2,((2,1),(1))}$, here the stable
dimensions are $1,8,97,3267,190139,17122837,2159496487,\ldots$ while the bounds
would be the values in column $l=3$ in table \ref{tab:mixed}.

By the definition of the inverse limit, $I_{k,(\lambda_1,\ldots,\lambda_k)}$
comes equipped with surjections onto the algebra of invariant polynomials over
$(\mathbb{S}_{\lambda_1}\mathbb{C}^{n_1})\otimes\cdots\otimes(\mathbb{S}_{\lambda_k}\mathbb{C}^{n_k})$
for every choice of the dimensions $n_1,\ldots,n_k$. Elements of the
inverse limit can therefore be directly interpreted as polynomial
invariants, although some of these will coincide when projected onto the
smaller algebras.

In the case of systems of bosons and fermions we have also given a combinatorial
description of an algebraically independent generating set of the inverse limit
in terms of regular graphs, showing in particular that the inverse limit is free
in these cases. This generating set can be interpreted as a set of polynomials
distinguishing between physically different types of nonlocal behaviour which is
minimal in the sense that any polynomial in its elements is nonzero provided the
single-particle Hilbert spaces are large enough.

We belive that these invariants will become a useful tool in the exploration of
genuine multiparticle quantum correlations in quantum mechanical systems containing
identical particles.

\section{The proofs}\label{sec:proofs}
\begin{proof}[Proof of lemma \ref{lem:lowdimiso}]
For every $n\in\mathbb{N}^k$ there is a surjection
$\mathbb{C}(S_{ml_1}\times\cdots\times S_{ml_k})\to(\End((\mathbb{C}^{n_1})^{\otimes ml_1}\otimes\cdots\otimes(\mathbb{C}^{n_k})^{\otimes ml_k}))^{LU_n}\simeq (S^{2m}((\mathbb{C}^{n_1})^{\otimes ml_1}\otimes\cdots\otimes(\mathbb{C}^{n_k})^{\otimes ml_k}\oplus ((\mathbb{C}^{n_1})^{\otimes ml_1}\otimes\cdots\otimes(\mathbb{C}^{n_k})^{\otimes ml_k}))^*)^{LU_n}$
defined by eq. \ref{eq:tensorinv}. Composing with the surjective map $(S^{2m}((\mathbb{C}^{n_1})^{\otimes ml_1}\otimes\cdots\otimes(\mathbb{C}^{n_k})^{\otimes ml_k}\oplus ((\mathbb{C}^{n_1})^{\otimes ml_1}\otimes\cdots\otimes(\mathbb{C}^{n_k})^{\otimes ml_k}))^*)^{LU_n}\twoheadrightarrow S(\mathcal{H}_{k,(\lambda_1,\ldots,\lambda_k),n}\oplus\mathcal{H}_{k,(\lambda_1,\ldots,\lambda_k),n}^*)^{LU_n}=I_{k,(\lambda_1,\ldots,\lambda_k),n}$
we obtain a map $\mathbb{C}(S_{ml_1}\times\cdots\times S_{ml_k})\to I_{k,(\lambda_1,\ldots,\lambda_k),n}$ that is onto the
degree $m$ homogeneous subspace, $I_{k,(\lambda_1,\ldots,\lambda_k),n}^{(m)}$. For $n\le n'$ we have the following commutative diagram:
\begin{equation}
\xymatrix{\mathbb{C}(S_{ml_1}\times\cdots\times S_{ml_k}) \ar@{->>}[r]\ar@{->>}[d]  &  I_{k,(\lambda_1,\ldots,\lambda_k),n}^{(m)}  \\
             I_{k,(\lambda_1,\ldots,\lambda_k),n'}^{(m)} \ar[ur]_{\varrho_{n,n'}}}
\end{equation}
This implies that the restriction of $\varrho_{n,n'}$ to the degree $m$ homogeneous subspace must also be surjective.
But if $(ml_1,\ldots,ml_k)\le n$, then $\dim I_{k,(\lambda_1,\ldots,\lambda_k),n}^{(m)}=\dim I_{k,(\lambda_1,\ldots,\lambda_k),n'}^{(m)}$ and therefore $\varrho_{n,n'}$ restricted to the degree $m$ homogeneous subspace is an isomorphism.
\end{proof}

Before proving lemma \ref{lem:wrindmany} we introduce some notation and summarize
some properties of wreath products following ref. \cite{Macdonald}. For a set $X$ let us denote the set of
partition valued functions on $X$ by $P(X)$, and let
\begin{equation}
P_m(X)=\{\varrho\in P(X)|\sum_x\in X|\varrho(x)|=m\}
\end{equation}
For a finite group $G$ the set of conjugacy classes of $G$ will be denoted by $G_*$.
Recall that the set of conjugacy classes of $G\wr S_m$ for any finite group $G$
is in bijection with $P_m(G_*)$ as follows: if $x=((g_1,g_2,\ldots,g_m),\pi)$ is an
element of $G\wr S_m$ then $\pi$ can be written as a product of disjoint cycles,
and for any cycle $(i_1,i_2,\ldots,i_k)$ we can form the product $g_1g_2\cdots g_k$
which is determined up to conjugation. The class of $x$ is then labelled by the
function $\varrho\in P_m(G_*)$ for which the number of $k$-s in $\varrho(c)$ is
the number of $k$-cycles of $\pi$ such that the corresponding cycle-product is in $c$.

The order of the centralizer of an element $x$ with type $\varrho$ is
\begin{equation}
Z_\varrho=\prod_{c\in G_*}z_{\varrho(c)}\zeta_c^{l(\varrho(c))}
\end{equation}
where $l$ denotes the length of the partition as usual and $\zeta_c$ is the order
of the centralizer in $G_*$ of an element in the conjugacy class $c$.

If $\gamma$,$\chi$ are characters of $G$ and $S_m$ respectively, then the value of
the character $\gamma\wr\chi$ on an element with type $\varrho$ is
\begin{equation}
\left(\prod_{c\in G_*}\gamma(c)^{l(\varrho(c))}\right)\chi(\lambda)
\end{equation}
where $\lambda=\cup_{c\in G_*}\varrho(c)$.

Now we prove the following observation:
\begin{lem}
Let $A,B$ be two finite groups, $m\in\mathbb{N}$, $\mu\vdash m$ and let $\alpha$, $\beta$ be class functions on
$A$ and $B$, respectively. Then
\begin{equation}
\ind_{(A\times B)\wr S_m}^{(A\wr S_m)\times(B\wr S_m)}(\alpha\times\beta)\wr\chi_{\mu}=\sum_{\mu_1,\mu_2\vdash m}\langle\chi_{\mu},\chi_{\mu_1}\chi_{\mu_2}\rangle_{S_m}(\alpha\wr\chi_{\mu_1})\times(\beta\wr\chi_{\mu_2})
\end{equation}
\end{lem}
\begin{proof}
We will prove that the inner product of the left side with the characteristic function
of any conjugacy class is the same as that of the right side.

A pair $(\varrho_A,\varrho_B)\in P_m(A_*)\times P_m(B_*)$ can therefore be identified
with a conjugacy class in $H':=(A\wr S_m)\times(B\wr S_m)$. The partitions
$\lambda_A=\cup_{a\in A_*}\varrho_A(a)$ and $\lambda_B=\cup_{b\in B_*}\varrho_B(b)$
depend  only on the conjugacy class of the two permutations in an element of type
$(\varrho_A,\varrho_B)$. Denoting the corresponding characteristic function by
$1_{(\varrho_A,\varrho_B)}$ and similarly for conjugacy classes of $A\wr S_m$,$B\wr S_m$
and $S_m$ we can write
\begin{equation}\label{eq:rhs}
\begin{split}
 & \langle 1_{(\varrho_A,\varrho_B)}, \sum_{\mu_1,\mu_2\vdash m}\langle\chi_\mu,\chi_{\mu_1}\chi_{\mu_2}\rangle_{S_m}(\alpha\wr\chi_{\mu_1})\times(\beta\wr\chi_{\mu_2})_{H'}  \\
 & = \sum_{\mu_1,\mu_2\vdash m}\langle\chi_\mu,\chi_{\mu_1}\chi_{\mu_2}\rangle_{S_m}\langle 1_{\varrho_A},\alpha\wr\chi_{\mu_1}\rangle_{A\wr S_m}\langle 1_{\varrho_B},\beta\wr\chi_{\mu_2}\rangle_{B\wr S_m}  \\
 & = \frac{z_{\lambda_A}\prod_{a\in A_*}\alpha(a)^{l(\varrho_A(a))}}{Z_{\varrho_A}}
     \frac{z_{\lambda_B}\prod_{b\in B_*}\beta (b)^{l(\varrho_B(b))}}{Z_{\varrho_B}}\sum_{\mu_1,\mu_2\vdash m}\langle\chi_\mu,\chi_{\mu_1}\chi_{\mu_2}\rangle\langle\chi_{\mu_1},\varphi_{\lambda_A}\rangle\langle\chi_{\mu_2},\varphi_{\lambda_B}\rangle  \\
 & = \frac{z_{\lambda_A}\prod_{a\in A_*}\alpha(a)^{l(\varrho_A(a))}}{Z_{\varrho_A}}
     \frac{z_{\lambda_B}\prod_{b\in B_*}\beta (b)^{l(\varrho_B(b))}}{Z_{\varrho_B}}\sum_{\mu_2\vdash m}\langle\chi_\mu\chi_{\mu_2},\varphi_{\lambda_A}\rangle\langle\chi_{\mu_2},\varphi_{\lambda_B}\rangle  \\
 & = \frac{z_{\lambda_A}\prod_{a\in A_*}\alpha(a)^{l(\varrho_A(a))}}{Z_{\varrho_A}}
     \frac{z_{\lambda_B}\prod_{b\in B_*}\beta (b)^{l(\varrho_B(b))}}{Z_{\varrho_B}}\langle\varphi_{\lambda_A},\varphi_{\lambda_B}\rangle  \\
\end{split}
\end{equation}
where we have used that irreducible characters of $S_m$ as well as characteristic functions are real.
The last inner product is $0$ if $\lambda_A\neq\lambda_B$ and $z_{\lambda_A}^{-1}\chi_\mu(\lambda_A)$ otherwise.

Now let us look at the left hand side. The value of the induced character at an
element $x\in H'$ is $0$ whenever $x$ is not conjugate to any element of
$H:=(A\times B)\wr S_m$ which happens precisely when the two permutations of $x$ are
not conjugate to each other, that is when $\lambda_A\neq\lambda_B$.

If $\lambda_A=\lambda_B=\lambda$ then we have
\begin{equation}\label{eq:lhs}
\begin{split}
\langle 1_{(\varrho_A,\varrho_B)},\ind_H^{H'}(\alpha\times\beta)\wr\chi_\mu\rangle_{H'}
 & = \langle\res_{H'}^H 1_{(\varrho_A,\varrho_B)},(\alpha\times\beta)\wr\chi_\mu\rangle_H  \\
 & = \sum_\varrho\langle 1_\varrho,(\alpha\times\beta)\wr\chi_\mu\rangle_H  \\
 & = \sum_\varrho\frac{1}{Z_\varrho}\left(\prod_{(a,b)\in A_*\times B_*}(\alpha(a)\beta(b))^{l(\varrho(a,b))}\right)\chi_\mu(\lambda)  \\
 & = \chi_\mu(\lambda)\prod_{a\in A_*}\alpha(a)^{l(\varrho_A(a))}\prod_{b\in B_*}\alpha(b)^{l(\varrho_B(b))}\sum_\varrho\frac{1}{Z_\varrho}
\end{split}
\end{equation}
where the sums are over those $\varrho\in P_m(A_*\times B_*)$ for which
$\cup_{b\in B_*}\varrho(a,b)=\varrho_A(a)$ and $\cup_{a\in A_*}\varrho(a,b)=\varrho_B(b)$ holds.

Comparing eq. (\ref{eq:lhs}) with eq. (\ref{eq:rhs}) one can see that we need to prove
that
\begin{equation}
\frac{1}{z_\lambda}\sum_\varrho\frac{1}{Z_\varrho}=\frac{1}{Z_{\varrho_A}}\frac{1}{Z_{\varrho_B}}
\end{equation}
Multiplying both sides by $|A|^m m!|B|^m m!$ we have on the left hand side the size of the
conjugacy class $\lambda$ in $S_m$ times the number of $x\in H$ whose image in $H'$ has
type $(\varrho_A,\varrho_B)$ and on the right hand side size of the conjugacy class $(\varrho_A,\varrho_B)$
in $H'=(A\wr S_m)\times(B\wr S_m)$.

An element $(a,\pi_1,b,\pi_2)$ of $H'$ can be uniquely written as the product
of an element in $H$ and one in $S_m\simeq\{e_{A^m}\}\times\{e\}\times\{e_B^m\}\times S_m\le H$ as follows:
$(a,\pi_1,b,\pi_1)(e_{A^m},e,e_{B^m},\pi_1^{-1}\pi_2)$, uniqueness follows from the
fact that the intersection of the two subgroups consists of only the identity.
Therefore every element of $H$ can be reached as the conjugate of some element in $H'$
with an element of the form $\hat{\sigma}=(e_{A^m},e,e_{B^m},\sigma)$ where $\sigma\in S_m$

It follows that we have a surjection $H\times S_m\to H'$ whose appropriate
restrictions
\begin{equation}
\{x\in H|\textrm{the type of $x$ in $H'$ is $(\varrho_A,\varrho_B)$}\}\times S_m\to\{y\in H'|\textrm{$y$ has type $(\varrho_A,\varrho_B)$}\}
\end{equation}
are also surjections.

Finally, $\hat{\sigma}(a,\pi_1,b,\pi_2)\hat{\sigma}^{-1}=(a,\pi_1,\sigma(b),\sigma\pi_2\sigma^{-1})\in H$
iff $\pi_1=\sigma\pi_2\sigma^{-1}$ iff $\sigma\in\tilde{\pi}Z_{S_m}(\pi_2)$ for a fixed $\pi$ such
that $\pi\pi_2\pi^{-1}=\pi_1$, which implies that the inverse image of any element in $H'$ has
precisely $|Z_{S_m}(\pi_2)|=z_\lambda$ elements, finishing the proof.
\end{proof}

Now we extend the above result to $k$ factors instead of just two:
\begin{proof}[Proof of lemma \ref{lem:wrindmany}]
We prove by induction using the $k=2$ case in the induction step. The $k=1$ case is easy to check.
Assuming the statement to be true for $1,2,\ldots,k-1$ we can write
\begin{equation}
\begin{split}
 & \ind_{(A_1\times\cdots\times A_k)\wr S_m}^{(A_1\wr S_m)\times\cdots\times(A_k\wr S_m)}(\alpha_1\times\cdots\times\alpha_k)\wr\chi_{\mu} \\
 & = \ind_{(A_1\times\cdots\times A_{k-1})\wr S_m\times A_k\wr S_m}^{(A_1\wr S_m)\times\cdots\times(A_k\wr S_m)}\ind_{(A_1\times\cdots\times A_k)\wr S_m}^{(A_1\times\cdots\times A_{k-1})\wr S_m\times A_k\wr S_m}(\alpha_1\times\cdots\times\alpha_k)\wr\chi_{\mu}  \\
 & = \sum_{\mu'_{k-1},\mu_k\vdash m}\langle\chi_\mu,\chi_{\mu'_{k-1}}\chi_{\mu_k}\rangle\left(\ind_{(A_1\times\cdots\times A_{k-1})\wr S_m}^{(A_1\wr S_m)\times\cdots\times(A_{k-1}\wr S_m)}(\alpha_1\times\cdots\times\alpha_{k-1})\wr\chi_{\mu'_{k-1}}\right)\times(\alpha_k\wr\chi_{\mu_k})  \\
 & = \sum_{\mu'_{k-1},\mu_k\vdash m}\langle\chi_\mu\chi_{\mu_k},\chi_{\mu'_{k-1}}\rangle\sum_{\mu_1,\ldots,\mu_{k-1}\vdash m}\langle\chi_{\mu'_{k-1}},\chi_{\mu_1}\cdots\chi_{\mu_{k-1}} \rangle(\alpha_1\wr\chi_{\mu_1})\times\cdots\times(\alpha_{k-1}\wr\chi_{\mu_{k-1}})  \\
 & = \sum_{\mu_1,\ldots,\mu_k\vdash m}\langle\chi_{\mu},\chi_{\mu_1}\cdots\chi_{\mu_k}\rangle_{S_m}(\alpha_1\wr\chi_{\mu_1})\times\cdots\times(\alpha_k\wr\chi_{\mu_k})
\end{split}
\end{equation}
using that irreducible characters of $S_m$ are real and form an orthonormal basis.
\end{proof}

\begin{proof}[Proof of lemma \ref{lem:product}]
As the map in eq. (\ref{eq:fdef}) defining $f_{[\sigma_1,\ldots,\sigma_k]}$ factors through
$S((\mathbb{C}^{n_1})^{\otimes l_1})\otimes\cdots\otimes(\mathbb{C}^{n_k})^{\otimes l_k})\oplus((\mathbb{C}^{n_1})^{\otimes l_1})\otimes\cdots\otimes(\mathbb{C}^{n_k})^{\otimes l_k}))^*$, we can also work in the symmetric algebra $(S((\mathbb{C}^{n_1})^{\otimes l_1})\otimes\cdots\otimes(\mathbb{C}^{n_k})^{\otimes l_k})\oplus((\mathbb{C}^{n_1})^{\otimes l_1})\otimes\cdots\otimes(\mathbb{C}^{n_k})^{\otimes l_k}))^*)^{LU_n}$ for
some large $n$. In this algebra the image of $(\sigma_1,\ldots,\sigma_k)$ is
\begin{multline}
\sum(e_{i_1^1}\otimes\cdots\otimes e_{i_1^{l_1}}\otimes\cdots\otimes e_{i_k^{1}}\otimes\cdots\otimes e_{i_k^{l_k}})\cdots \\ \cdots(e_{i_1^{(m-1)l_1+1}}\otimes\cdots\otimes e_{i_1^{ml_1}}\otimes\cdots\otimes e_{i_k^{(m-1)l_k+1}}\otimes\cdots\otimes e_{i_k^{ml_k}})\cdot  \\
\cdot(e_{i_1^{\sigma_1(1)}}\otimes\cdots\otimes e_{i_1^{\sigma_1(l_1)}}\otimes\cdots\otimes e_{i_k^{\sigma_k(1)}}\otimes\cdots\otimes e_{i_k^{\sigma_k(l_k)}})\cdots \\ \cdots(e_{i_1^{\sigma_1((m-1)l_1+1)}}\otimes\cdots\otimes e_{i_1^{\sigma_1(ml_1)}}\otimes\cdots\otimes e_{i_k^{\sigma_k((m-1)l_k+1)}}\otimes\cdots\otimes e_{i_k^{\sigma_k(ml_k)}})
\end{multline}
where the sum is over the possible values of the indices $i_1^1,\ldots,i_1^{ml_1},\ldots,i_k^1,\ldots,i_k^{ml_k}$
and similarly for $(\sigma'_1,\ldots,\sigma'_k)$. It is convenient to denote the indices
in the sum corresponding to $(\sigma'_1,\ldots,\sigma'_k)$ by $i_1^{ml_1+1},\ldots,i_1^{(m+m')l_1},\ldots,i_k^{ml_k+1},\ldots,i_k^{(m+m')l_k}$
and to regard the permutation $\sigma'_j$ as a bijection from $\{ml_j+1,ml_j+2,\ldots,(m+m')l_j\}$
to itself. This convention clearly does not affect the definition of $f_{[\sigma'_1,\ldots,\sigma'_k]}$
and it is consistent with the definition of $\star$.
Then we have
\begin{equation}
\begin{split}
\sum & (e_{i_1^1}\otimes\cdots\otimes e_{i_1^{l_1}}\otimes\cdots\otimes e_{i_k^{1}}\otimes\cdots\otimes e_{i_k^{l_k}})\cdots \\
& \cdots(e_{i_1^{(m-1)l_1+1}}\otimes\cdots\otimes e_{i_1^{ml_1}}\otimes\cdots\otimes e_{i_k^{(m-1)l_k+1}}\otimes\cdots\otimes e_{i_k^{ml_k}})\cdot  \\
& \cdot(e_{i_1^{\sigma_1(1)}}\otimes\cdots\otimes e_{i_1^{\sigma_1(l_1)}}\otimes\cdots\otimes e_{i_k^{\sigma_k(1)}}\otimes\cdots\otimes e_{i_k^{\sigma_k(l_k)}})\cdots \\
& \cdots(e_{i_1^{\sigma_1((m-1)l_1+1)}}\otimes\cdots\otimes e_{i_1^{\sigma_1(ml_1)}}\otimes\cdots\otimes e_{i_k^{\sigma_k((m-1)l_k+1)}}\otimes\cdots\otimes e_{i_k^{\sigma_k(ml_k)}})\cdot \\
\cdot\sum &(e_{i_1^{ml_1+1}}\otimes\cdots\otimes e_{i_1^{ml_1+l_1}}\otimes\cdots\otimes e_{i_k^{ml_k+1}}\otimes\cdots\otimes e_{i_k^{ml_k+l_k}})\cdots \\
& \cdots(e_{i_1^{ml_1+(m'-1)l_1+1}}\otimes\cdots\otimes e_{i_1^{ml_1+m'l_1}}\otimes\cdots\otimes e_{i_k^{ml_k+(m'-1)l_k+1}}\otimes\cdots\otimes e_{i_k^{ml_k+m'l_k}})\cdot  \\
& \cdot(e_{i_1^{\sigma_1(ml_1+1)}}\otimes\cdots\otimes e_{i_1^{\sigma_1(ml_1+l_1)}}\otimes\cdots\otimes e_{i_k^{\sigma_k(ml_k+1)}}\otimes\cdots\otimes e_{i_k^{\sigma_k(ml_k+l_k)}})\cdots \\
& \cdots(e_{i_1^{\sigma_1((m+m'-1)l_1+1)}}\otimes\cdots\otimes e_{i_1^{\sigma_1((m+m')l_1)}}\otimes\cdots\otimes e_{i_k^{\sigma_k((m+m'-1)l_k+1)}}\otimes\cdots\otimes e_{i_k^{\sigma_k((m+m')l_k)}}) = \\
= \sum(e_{i_1^1} & \otimes\cdots\otimes e_{i_1^{l_1}}\otimes\cdots\otimes e_{i_k^{1}}\otimes\cdots\otimes e_{i_k^{l_k}})\cdots \\
& \cdots(e_{i_1^{((m+m')-1)l_1+1}}\otimes\cdots\otimes e_{i_1^{(m+m')l_1}}\otimes\cdots\otimes e_{i_k^{((m+m')-1)l_k+1}}\otimes\cdots\otimes e_{i_k^{(m+m')l_k}})\cdot  \\
& \cdot(e_{i_1^{\sigma''_1(1)}}\otimes\cdots\otimes e_{i_1^{\sigma''_1(l_1)}}\otimes\cdots\otimes e_{i_k^{\sigma''_k(1)}}\otimes\cdots\otimes e_{i_k^{\sigma''_k(l_k)}})\cdots \\
& \cdots(e_{i_1^{\sigma''_1(((m+m')-1)l_1+1)}}\otimes\cdots\otimes e_{i_1^{\sigma''_1((m+m')l_1)}}\otimes\cdots\otimes e_{i_k^{\sigma''_k(((m+m')-1)l_k+1)}}\otimes\cdots\otimes e_{i_k^{\sigma''_k((m+m')l_k)}})
\end{split}
\end{equation}
where $\sigma''_j$ is the image of $(\sigma_j,\sigma'_j)$ under the inclusion
$S_m\times S_m'\hookrightarrow S_{m+m'}$. But by the definition of $\star$ we
have $[\sigma''_1,\ldots,\sigma''_k]=[\sigma_1,\ldots,\sigma_k]\star[\sigma'_1,\ldots,\sigma'_k]$.
\end{proof}

\begin{proof}[Proof of lemma \ref{lem:vanishing}]
The sum in eq. (\ref{eq:fdef}) can be rewritten as a nested sum, grouping together
the possible tuples of indices such that the sets of multisets $\{\{i_1^1,\ldots,i_1^{|\lambda_1|}\},\{i_1^{|\lambda_1|+1},\ldots,i_1^{2|\lambda_1|}\},$ \ldots $,\{i_1^{(m-1)|\lambda_1|+1},\ldots,i_1^{m|\lambda_1|}\}\},$\ldots $,\{\{i_k^1,\ldots,i_k^{|\lambda_k|}\},\ldots,\{i_k^{(m-1)|\lambda_k|+1},\ldots,i_k^{m|\lambda_k|}\}\}$ and also the sets $\{\{i_1^{\sigma_1(1)},\ldots,i_1^{\sigma_1(|\lambda_1|)}\},\{i_1^{\sigma_1(|\lambda_1|+1)},\ldots,i_1^{\sigma_1(2|\lambda_1|)}\},\ldots,\{i_1^{\sigma_1((m-1)|\lambda_1|+1)},\ldots,i_1^{\sigma_1(m|\lambda_1|)}\}\},$ \ldots $,\{\{i_k^1,\ldots$ $,i_k^{\sigma_1(|\lambda_k|)}\},\ldots,\{i_k^{\sigma_1((m-1)|\lambda_k|+1)},\ldots,i_k^{\sigma_1(m|\lambda_k|)}\}\}$ are kept fixed.
The inner sum is therefore over permutations stabilizing the above structure which is
precisely ${H_m}_s$, since ${H_m}$ and $s{H_m}s^{-1}$ is the stabilizer of the first and last $k$ sets,
respectively.

Taking into account the sign changes introduced when flipping a pair in a wedge product
we have that the inner sums are proportional to (the projection of)
\begin{equation}
\begin{split}
& \sum_{s\in {H_m}_s}\overline{((\chi_{\lambda_1}\times\cdots\times\chi_{\lambda_k})\wr 1)^s(h)}((\chi_{\lambda_1}\times\cdots\times\chi_{\lambda_k})\wr 1)(h)\cdot \\
 & \cdot \left[(e_{i_1^1}\otimes\cdots\otimes e_{i_1^{|\lambda_1|}})\otimes\cdots\otimes(e_{i_1^{\sigma_1(1)}}^*\otimes\cdots\otimes e_{i_1^{\sigma_1(|\lambda_1|)}}^*)\otimes\cdots\right]= \\
= & [\ldots]\langle((\chi_{\lambda_1}\times\cdots\times\chi_{\lambda_k})\wr 1)^s,\res_{H_m}^{{H_m}_s}(\chi_{\lambda_1}\times\cdots\times\chi_{\lambda_k})\wr 1\rangle_{{H_m}_s}
\end{split}
\end{equation}
and therefore vanish by assumption.
\end{proof}

\end{document}